\documentclass[
]{ceurart}

\usepackage{listings}
\usepackage{multicol}
\usepackage{amsmath}
\usepackage{amsthm}
\usepackage{amssymb}
\usepackage{enumitem}
\usepackage{amsfonts}
\usepackage{mathtools}
\usepackage{xcolor}
\usepackage{hyperref}
\usepackage{cleveref}
\usepackage{tikz}
\usepackage{thmtools}
\usepackage{thm-restate}
\usetikzlibrary{arrows,calc,patterns,positioning,shapes,fit}
\usetikzlibrary{decorations.pathmorphing}
\tikzset{
modal/.style={>=stealth,shorten >=1pt,shorten <=1pt,auto,
node distance=1.5cm,semithick},
world/.style={circle,opacity=0,draw=white,minimum size=1cm,fill=white},
point/.style={circle,draw,fill=black,inner sep=0.5mm},
reflexive/.style={->,in=120,out=60,loop,looseness=#1},
reflexive/.default={5},
reflexive point/.style={->,in=135,out=45,loop,looseness=#1},
reflexive point/.default={25},
coil/.style={decorate, decoration={coil,amplitude=4pt,segment length=5pt}},
snake/.style={decorate, decoration={snake}},
zigzag/.style={decorate, decoration={zigzag}}
}

\newtheorem{thm}{Theorem}[section]

\newtheorem{definition}[thm]{Definition} 
\newtheorem{example}[thm]{Example} 

\newcommand{\ThreeSAT}{\mathrm{3SAT}}
\newcommand{\MaxTwoSAT}{\mathrm{Max2SAT}}

\newcommand{\NP}{\mathrm{NP}}
\newcommand{\NL}{\mathrm{NL}}
\newcommand{\LOGSPACE}{\mathrm{LOGSPACE}}

\def\dep{\mathord{=}}

\begin{document}

%%
%% Rights management information.
%% CC-BY is default license.
% \copyrightyear{2026}
% \copyrightclause{Copyright for this paper by its authors.
%   Use permitted under Creative Commons License Attribution 4.0
%   International (CC BY 4.0).}

%%
%% This command is for the conference information
\conference{}
%LINDA'26,
%  July 24, 2026, Lisbon, Portugal}

%%
%% The "title" command
\title{Approximate Functional Dependencies---Implication Problem Revisited}

% \tnotemark[1]
% \tnotetext[1]{You can use this document as the template for preparing your
%   publication. We recommend using the latest version of the ceurart style.}

%%
%% The "author" command and its associated commands are used to define
%% the authors and their affiliations.
\author[1]{Nicolas Fröhlich}[%
orcid=0009-0003-5413-1823,
email=nicolas.froehlich@thi.uni-hannover.de,
]
%\cormark[1]
\fnmark[1]
\address[1]{Leibniz Universität Hannover, Germany}
% \address[2]{Joint Institute for Nuclear Research,
%   6 Joliot-Curie, Dubna, Moscow region, 141980, Russian Federation}

\author[2]{Matilda Häggblom}[%
orcid=0009-0003-6289-7853,
email=matilda.haggblom@helsinki.fi,
]
\fnmark[1]
\address[2]{University of Helsinki, Finland}

\author[2]{\AA sa Hirvonen}[%
orcid=0000-0003-2149-4153,
email=asa.hirvonen@helsinki.fi,
]
\fnmark[1]
% \address[4]{University of Skövde, Högskolevägen 1, 541 28 Skövde, Sweden}

\author[1]{Minna Hirvonen}[%
orcid=0000-0002-2701-9620,
email=minna.hirvonen@thi.uni-hannover.de,
]
\fnmark[1]

%% Footnotes
% \cortext[1]{Corresponding author.}
\fntext[1]{These authors contributed equally.}

%%
%% The abstract is a short summary of the work to be presented in the
%% article.
\begin{abstract}
  %A clear and well-documented \LaTeX{} document is presented as an
  %article formatted for publication by CEUR-WS in a conference
  %proceedings. Based on the ``ceurart'' document class, this article
  %presents and explains many of the common variations, as well as many
  %of the formatting elements an author may use in the preparation of
 %the documentation of their work.\\~\\
Functional dependencies are an important and well-studied class of database constraints that correspond to a notion expressed by dependence atoms in team logic. In practice, data often contain errors, so in some cases it might be useful to allow the database to have a small number of tuples that violate the desired dependency. Väänänen (2017)
%\cite{Väänänen2017} 
studied the axiomatisation of a notion of approximate dependence that specifies for each dependence atom how much of the database can be disregarded. We demonstrate that the interaction of approximate dependence atoms is more complicated than previously thought in the sense that there is a semantic consequence that is not captured by the inference rules introduced before. We show that Väänänen's axiomatisation is still complete in the restricted case of unary dependencies. We also consider the complexity of model checking for approximate dependence: it is NP-complete for disjunctions of two atoms and LOGSPACE-hard for individual atoms.
\end{abstract}

%%
%% Keywords. The author(s) should pick words that accurately describe
%% the work being presented. Separate the keywords with commas.

%%
%% This command processes the author and affiliation and title
%% information and builds the first part of the formatted document.
\maketitle

% \section{LINDA abstract}

% +++ Deadline April 25. Extended abstract 2-5 pages long, excluding references. Appendix and link to preprint can be included. 

% \matilda{Should we try to write some short version of the proofs to the theorem and lemma, or simply state the claims in the abstract and include the proofs in the appendix?}
% \minna{I think that either is fine, depends on how much space we have left after adding the intro?}

\section{Introduction}

Functional dependencies were introduced in database theory by Codd in the 1970's \cite{Codd71}, and axiomatised by Armstrong \cite{Armstrong1974} a few years later. In the unirelational case they are equivalent to dependence atoms in team logic. Teams were originally defined by Hodges \cite{hodges1997,hodges19972} to give a compositional semantics for so-called independence-friendly logic. Väänänen \cite{vaananen2007dependence} introduced dependence logic, allowing the explicit study of relationships between variables. 

The rise of big data offers strong motivation to study dependencies that disregard statistical outliers in the data set. Approximate notions of dependence are also motivated by considerations of natural language, such as `most', `all but a few', and were formalised for functional dependencies by Kivinen and Mannila \cite{KIVINEN1995129} with various measures of the part of the information allowed to be erroneous. Väänänen adapted one of the definitions as approximate dependence atoms in the team semantic setting \cite{Väänänen2017}, and proposed an axiomatisation of its implication problem.

% Teams were originally defined by Hodges [ADD CITATION] to give a compositional semantics for so-called independence friendly logic. Väänänen \cite{vaananen2007dependence} introduced dependence atoms, allowing to explicitly study relationships between variables. Later Väänänen \cite{Väänänen2017} defined an approximate notion of dependence, motivated by considerations of natural language, such as 'most', 'all but a few', where he studies a notion of dependence, where part of the information is allowed to be erroneous. 

%The rise of big data offers strong motivation to study dependencies that disregard statistical outliers in the data set. Väänänen's paper \cite{Väänänen2017} defines a notion of approximate dependence, where the atoms specify how much of the team/database needs to be disregarded in order for the dependence relation to hold. 

Unfortunately, the axiomatisation in \cite{Väänänen2017} disregards a subtle point in dependence that only turns up in the approximate notion. Where in pure functional dependencies, a tuple $yz$ depends on $x$ exactly when both $y$ and $z$ depend on $x$, in an approximate version, where small errors are allowed, the distribution of errors between variables introduces interplay between the variables that is hard to axiomatise. In this paper, we illustrate this phenomenon by presenting a `missing rule' - a semantic consequence that cannot be obtained from the axiomatisation in \cite{Väänänen2017}.

Väänänen's axiomatisation is still complete in the restricted case of \emph{unary dependencies}, where simple dependence statements `$y$ depends on $x$' are considered for single variables $x$ and $y$. We extract the necessary rules for unary and constant dependencies, and prove completeness (there was an error in the proof in \cite{Väänänen2017}). Such axiomatizations could help speed up algorithms for discovering approximate dependencies in databases (see, e.g., \cite{Huhtala1999}), with applications in data mining. 

We also examine the complexity of the model checking problem for approximate dependence.

\section{Definitions}

We recall basic team semantic definitions together with the approximate dependence atoms from \cite{Väänänen2017}.

%\minna{Added some definitions below. It might be that we don't need all of them, so we can remove/shorten these if needed.}
Let $D$ and $M$ be sets of variables and values, respectively. An assignment $s$ over $D$ is a function $s\colon D\to M$. A team $X$ over $D$ is a set of assignments $s\colon D\to M$. For a team $X$ over $D$ and a tuple $x=x_1\dots x_n$ of variables from D, we define the notation $X[x]:=\{s(x)\in M^n\mid s\in X\}$, where $s(x)=s(x_1)\dots s(x_n)$. For $x=x_1\dots x_n$ and $y=y_1\dots y_m$, we write $xy$ for the tuple $x_1\dots x_ny_1\dots y_m$. The notation $|x|$ refers to the length of the tuple $x$, e.g, $|x|=n$ for $x=x_1\dots x_n$. 

Let $x$ and $y$ be finite tuples of variables from $D$ such that $y$ is nonempty, i.e., $|y|\geq 1$. An expression $\dep(x,y)$ is called a \emph{dependence atom}. The satisfaction relation between a team $X$ over $D$ and the atom $\dep(x,y)$ is defined by %satisfied in team $X$, written as $X\models\dep(x,y)$, if and only if for all $s,s\in X$, $s(x)=s'(x)$ implies $s(y)=s'(y)$.
$$\text{$X\models\dep(x,y)$ if and only if for all $s,s'\in X$, $s(x)=s'(x)$ implies $s(y)=s'(y)$.}$$

Since $X\models\dep(x,y)$ if and only if there is a function $f\colon X[x]\to X[y]$ such that $f(s(x))=s(y)$ for all $s\in X$, dependence atoms are often called \emph{functional dependencies}, especially in database theory. Note that the tuple $x$ is allowed to be empty\footnote{In the dependence atom $\dep(x,y)$, the tuple $y$ is assumed to be nonempty, because otherwise the atom $\dep(x,y)$ would be trivially satisfied by any team over $D$ that contains the variables $x$.}, and in that case, we write $\dep(y)$ instead of $\dep(,y)$. The atom $\dep(y)$ is called the \emph{constancy atom}, as it states that the value of $y$ is constant in the team. Let $\Sigma\cup\{\tau\}$ be a set of atoms with variables from $D$. We write $X\models\Sigma$, if $X\models\sigma$ for all $\sigma\in\Sigma$. We say that $\Sigma$ \emph{semantically entails} $\tau$, written as $\Sigma\models\tau$, if for all $X$ over $D$, we have that $X\models\Sigma$ implies $X\models\tau$.

We consider a notion of \emph{approximate dependence} that states that a dependence must hold in a team when we are allowed to remove some assignments (usually a small portion of the whole team). 
\begin{definition}[\cite{Väänänen2017}]\label{Vään.def.non-restr2}
Let $p$ be a real number, $0 \leq p \leq 1$. For finite teams, $X\models {\dep_p(x, y)}$ if there is a subteam $Y\subseteq X$, $|Y | \leq p \cdot |X|$, such that $X \setminus Y\models \dep(x, y)$. %For arbitrary teams (finite or infinite)$X\models =^\prime (x, y)$ mod finite, if there is a finite $Y$ such that $X \setminus Y$ satisfies$=(x, y)$.
\end{definition}

The following set of rules is sound for approximate dependencies. 

\begin{definition}[\cite{Väänänen2017}] \label{def:OGrules} 
Rules $A1$-$A7$ form the system $\mathbf{A}$ for approximate dependencies.

% \begin{multicols}{2}
\begin{enumerate}
\centering
\begin{minipage}{0.375\textwidth}
\item[(A1)] $\dep_0(xy, x)$.

\item[(A2)]  $\dep_1(x,y)$.

\item[(A3)] If $\dep_p(x, yv)$, then $\dep_p(xu, y)$.

\item[(A4)] If $\dep_p(x, y)$, then $\dep_p(xu, yu)$.
\end{minipage}
\begin{minipage}{0.475\textwidth}
\item[(A5)] If $\dep_p(xu, yv)$, then $\dep_p(ux, yv)$ and $\dep_p(xu, vy)$.

\item[(A6)] For $r=min\{p+q,1\}$, if $\dep_p(x, y)$ and $\dep_q(y, v)$, then $\dep_r(x, v)$.

\item[(A7)] For $p \leq q \leq 1$,  if $\dep_p(x, y)$, then $\dep_q(x, y)$. 
\end{minipage}
\end{enumerate}  
% \end{multicols}
\end{definition}

For a set of approximate dependencies $\Sigma$ and a system of rules $\mathbf{B}$, we write $\Sigma\vdash_\mathbf{B}\dep_p(x,y)$ if we can apply the rules in $\mathbf{B}$ to the dependencies in $\Sigma$ to derive the conclusion $\dep_p(x,y)$. A system $\mathbf{B}$ is sound if whenever $\Sigma\vdash_\mathbf{B} \dep_p(x,y)$, also $\Sigma\models \dep_p(x,y)$. 
%\minna{Should we formally define soundness or at least the notation $\vdash$?} 
As shown in \cite{Väänänen2017}, the rules are sound, but we challenge the claim that they are \emph{complete} in the sense claimed in \cite{Väänänen2017}, i.e., that for a finite set of approximate dependence atoms, $\Sigma\models\dep_p(x, y)$ would imply that $\Sigma\vdash_{\mathbf{A}}\dep_p(x, y)$.
%
% \begin{quote}%(\cite{Väänänen2017})
%     Suppose $\Sigma$ is a finite set of approximate dependence atoms. Then
% $\dep_p(x, y)$ follows from $\Sigma$ by the above axioms and rules if and only if every finite
% team satisfying $\Sigma$ also satisfies $\dep_p(x, y)$.
% \end{quote}

First, let us note a minor issue that arises when the dependence is defined on sequences rather than on sets of variables/attributes, as is common in the database theory literature. Using system $\mathbf{A}$, we cannot derive $\dep_p(x,xy)$ from $\dep_p(x,y)$. Instead, we can derive $\dep_p(x,y)\vdash_\mathbf{A}\dep_p(xx,xy)$, so if we consider $xx$ as different from $x$, the set of rules is not complete even for usual dependencies with approximation $0$. %In particular, we can not derive $\vdash\dep_0(x,xx)$. 
To solve this, as pointed out in \cite{HannulaKontinenYang_LogicsTeamSemantics}, we can replace the rules $A3$ and $A4$ with the rule $A_{34}$: $$(A_{34})\quad \text{If $\dep_p(x,y)$, then $\dep_p(x,xy)$.}$$ 

System $\mathbf{A}^*$ consisting of rules $A1,A2,A_{34},A5,A6$ and $A7$ is complete for usual dependencies sensitive to repetition of variables (see e.g., \cite{Galliani2014} for a completeness proof in the team semantic setting). %It is now an easy exercise to show that the rules $A3$ and $A4$ are derivable using rules $A1,A2,A^*,A5,A6$ and $A7$, and that this set is complete for usual dependencies (see e.g., \cite{Galliani2014} for a completeness proof in the team semantic setting).
However, we show in the next section that there are semantic entailments specific to the approximate setting that are not derivable even in this modified set of rules $\mathbf{A}^*$. 

%\matilda{fix/change this}From here on, we simply write $\vdash$ instead of $\vdash_{A^*}$ unless otherwise specified. 

% \begin{definition}[\cite{Väänänen2017} with exact set of rules from DP-book-draft, some issue with repetitions for the set in \cite{Väänänen2017}; $\dep(x,y)\not\vdash\dep(x,xy)$, but only $\dep(x,y)\vdash\dep(xx,xy)$, so if we consider $xx\neq x$, it is not complete even for usual dependencies when we are sensitive to repetitions.]\label{def:rules}

% The complete set of rules for approximate dependence atoms are listed below.
%   \begin{enumerate} 
% \item[(D1)] $\dep_0(xy, x)$.

% \item[(D2)] For $r=min\{p+q,1\}$, if $\dep_p(y, z)$ and $\dep_q(z, x)$, then $\dep_r(y, x)$.

% \item[(D3)] If $\dep_p(x, y)$, then $\dep_p(x, xy)$.

% \item[(D4)] If $\dep_p(xu, yv)$, then $\dep_p(xu, vy)$ and $\dep_p(ux, yv)$.

% \item[(D5)]  $=_1\!(y, x)$.

% \item[(D6)] For $p \leq q \leq 1$,  if $\dep_p(x, y)$, then $\dep_q(x, y)$.
% \end{enumerate}  
% \end{definition}

\section{Missing Rule}\label{missing}

%\nicolas{write Åsa's idea}
%\nicolas{with or without constancy atom?} \matilda{I think since we include cases with constancy atoms in the corrected proof, it is fine to have them here too. (But if we can remove it that makes for a simpler rule of course...)}
%\asa{I don't think the idea works without constancy, so we need it.} 
We aim to show that there is a semantic entailment among approximate dependence atoms that is not derivable using rules from the system $\mathbf{A}^*$, proving that the system is incomplete. First, we exemplify this semantic entailment by examining a specific team.   

\begin{example}\label{ex:missing rule}
    Let $\Sigma = \{\dep_0(x), \dep_{\frac{1}{8}}(x, v_1v_2), \dep_{\frac{2}{8}}(x, v_1v_3), \dep_{\frac{2}{8}}(v_1, y), \dep_0(v_2v_3, y)\}$ be a set of approximate atoms.
    Consider the team $X$ depicted in \Cref{fig:missing rule}.
    We have that $X \models \Sigma$, while having to remove the maximum number of assignments in each approximation atom in $\Sigma$.
    Next, consider the dependence atom $\dep(x, y)$.
    For $\dep(x, y)$ to be true, five assignments in $X$ need to be removed; therefore $X \models \dep_{\frac{5}{16}}(x, y)$ holds.
    However $\dep_{\frac{5}{16}}(x, y)$ cannot be derived by the rules in system $\mathbf{A}^*$.
    The best approximation atom derivable is $\dep_{\frac{6}{16}}(x, y)$ via transitivity of $\dep_{\frac{1}{8}}(x, v_1v_2)$ and $\dep_{\frac{2}{8}}(v_1, y)$.
    While $X \models \dep_{\frac{6}{16}}(x, y)$ is obvious, the ratio $\frac{6}{16}$ is clearly not optimal.
%\minna{Removed spacing btwn the paragrphs in the ex above to save space.}
    \begin{table}
            \caption{Example team for \Cref{ex:missing rule}.
        The team satisfies all approximation atoms in $\Sigma$, while also satisfying $\dep_{\frac{5}{16}}(x, y)$.
        However $\dep_{\frac{5}{16}}(x, y)$ cannot be derived from $\Sigma$ given the rules in system $\mathbf{A}^*$.}
        \label{fig:missing rule}
        \[
            \begin{array}{ccccccc}
            \toprule
                 & x & v_1 & v_2 & v_3 & y & z_0 \\
            \midrule
                s_0 & 0 & 0 & 0 & 1 & 1 & 0 \\
                s_1 & 0 & 0 & 0 & 2 & 2 & 1 \\
                s_2 & 0 & 0 & 0 & 3 & 3 & 2 \\
                s_3 & 0 & 1 & 0 & 4 & 4 & 3 \\
                s_4 & 0 & 0 & 1 & 0 & 5 & 4 \\
                s_5 & 0 & 0 & 0 & 0 & 0 & 5 \\
                s_6 & 0 & 0 & 0 & 0 & 0 & 6 \\
                \vdots & \vdots & \vdots & \vdots & \vdots & \vdots & \vdots \\
                s_{15} & 0 & 0 & 0 & 0 & 0 & 15 \\
            \bottomrule
            \end{array}
        \]
    \end{table}
\end{example}

The next result shows that this example is not an isolated case, but that all teams satisfying $\Sigma$ must also satisfy $\dep_{\frac{5}{16}}(x, y)$.
\begin{restatable}{prop}{lemMissing}
    Let $\Sigma = \{\dep_0(x), \dep_{\frac{1}{8}}(x, v_1v_2), \dep_{\frac{2}{8}}(x, v_1v_3), \dep_{\frac{2}{8}}(v_1, y), \dep_0(v_2v_3, y)\}$ be a set of approximate atoms.
    If $X \models \Sigma$ for any team $X$, then $X \models \dep_{\frac{5}{16}}(x, y)$.
\end{restatable}
However, since $\Sigma \not\vdash_{\mathbf{A}^*}\, \dep_{\frac{5}{16}}(x, y)$,
we conclude that $\mathbf{A}^*$ is an incomplete set of rules.

It is worth noting that the missing rule does not only show incompleteness in the case with tuples as \emph{dependent} variables (i.e., in the right part of the atom). If one studies \emph{right unary} dependence atoms, i.e., atoms of the form $\dep_p(x_1 \dots x_n,y_1)$, one can simulate atoms of the form $\dep_p(x,y_1y_2)$ by introducing a new variable $z$ and using dependencies $\dep_0(z,y_1)$, $\dep_0(z,y_2)$, $\dep_0(y_1y_2,z)$, $\dep_p(x,z)$.

\section{Restricted Completeness Theorem}

% +++ Matilda writes 

% \matilda{Do we define that in general, the RHS must be a nonempty sequence of variables, maybe not? $\dep(x,\;)\equiv\top$, right? }

% \matilda{Should we use multiteams or an unrelated variable to the same effect?}
% \minna{Multiteams might be good, because then we can mention them also in the monoid-case? (Although there we have the "weighted" semantics instead of rations.)}

We prove the completeness of the original system in \Cref{def:OGrules} for a restricted set of approximate dependence atoms. Namely, when we restrict to rational approximations $p\in\mathbb{Q}\cap[0,1]$ and to unary and so-called 1-constancy atoms of the form $\dep_p(x,y)$, where $|x|\leq 1$ and $|y|=1$.

% We also restrict the proof system to unary and 1-constancy atoms in the form of $\mathbf{A}^1$ containing the arity restricted variants of rules

Due to the arity restrictions, we can omit rules $A4$ and $A5$. We denote the remaining rules restricted to unary and 1-constancy atoms by $A1^1$, $A2^1$, $A3^1$, $A6^1$ and $A7^1$ and obtain system $\mathbf{A}^1$. In particular, $A1^1$ takes the simple form `$\dep_0(x,x)$', and $A3^1$ is of the form `if $\dep_p(y)$, then $\dep_p(u,y)$'.

\begin{restatable}{thm}{thmUnary}\label{completeness}
     Let $\Sigma\cup \{\dep_p(x,y)\}$ be a finite set of unary and 1-constancy rational approximate dependence atoms. If  $\,\Sigma\models \dep_p(x,y)$, then $\Sigma\vdash_{\mathbf{A}^1} \dep_p(x,y)$.% using ++unary+++ rules $A1$, $A2$, $A3$, $A6$ and $A7$ +++ from \Cref{def:OGrules}. \label{R compl}
 \end{restatable} 

 % \asa{Should we remove $z_0$ in the statement of the theorem?}

 To prove completeness under these restrictions, we construct a counterexample team $X$ similar to the one in \cite{Väänänen2017}, but with some core differences, one of which we discuss next. Let $\Sigma\cup \{\dep_p(x,y)\}$ be as in the above Theorem. Our goal is to show, under the assumption $\Sigma\not\vdash_{\mathbf{A}^1}\dep_p(x,y)$, that there is a team $X$ that satisfies all atoms in $\Sigma$ but not $\dep_p(x,y)$.
 Like in \cite{Väänänen2017}, for each variable $u$ we find the smallest approximation $d(u)$ for which $\dep_{d(u)}(x,u)$ is derivable from $\Sigma$. Differing from \cite{Väänänen2017}, we can find the least common denominator $n-1$ for the approximations appearing in $\Sigma\cup \{\dep_p(x,y)\}$, since they are rational numbers. We then construct a team of size $n$ such that the approximations $d(u)=\frac{k}{n-1}$ are \emph{uniformly} encoded with a slightly better approximation, namely $\frac{k}{n}$. This type of uniformity is missing from \cite{Väänänen2017}, making that completeness proof erroneous even for unary atoms, possibly because of the attempt to cover real-numbered approximations.   % and subsequently an $\epsilon\in\mathbb{Q}\cap[0,1]$ such that for each variable $u$, the smallest approximation $r$ for which $X\models\dep_r(x,u)$ is $d(u)-\epsilon$. 
 Two examples of counterexample teams are presented in \Cref{R team}, one of which takes advantage of a dummy variable $z_0$ and the other does not. %. An alternative counterexample team without this dummy variable is presented in the same table in the form of $X'$. %This uniformity is missing in \cite{Väänänen2017} and causes an error in the proof. 
A similar construction for unary approximate \emph{inclusion atoms} with rational approximations also appeared in \cite{Haggblom2026}.

 % \begin{proof}
 %     sketch.
 % \end{proof}

\begin{table}[tb]
     \centering
      \caption{Suppose that $\dep_{\frac{1}{4}}(x,y)$ is not derivable from $\Sigma$ using rules from system $\mathbf{A}^1$, where $\Sigma$ is $\{\dep_{\frac{1}{4}}(x,w),\, \dep_{\frac{1}{4}}(w,y)\}$. Then $d(x)=0$, $d(w)=\frac{1}{4}$, and $d(y)=\frac{1}{2}$ and we can construct counterexample teams $X$ and $Y$ as illustrated.}
     \label{R team}
 \[
\begin{array}{c ccc cc}
\toprule 	
 X \quad & x  &w&  y & z_0 \\ 
\midrule 

s_0& 0 & 0& 0
  & 0
  \\
s_1&   0 & 1& 1
  & 1
  \\
s_2& 0 & 1& 2
 & 2
  \\
 s_3&  0 & 1&2
  & 3
  \\
s_4&  0 & 1& 2
 & 4
  \\
\bottomrule 
\end{array}
\hspace{2cm}
\begin{array}{c ccc c}
\toprule 	
 Y \quad & x  &w&  y  \\ 
\midrule 

s_0& 0 & 0& 0
  
  \\
s_1&   0 & 1& 1

  \\
s_2& 0 & 1& 2

  \\
 s_3&  a & a&a
  
  \\
s_4&  b& b& b

  \\
\bottomrule 
\end{array}
\] 
 \end{table}
 % To extend the completeness result further could include removing the arity restriction, allowing real-numbered approximations, and not relying on a dummy variable to prove the completeness theorem. As we have seen, allowing atoms $\dep_p(x,y)$, $|x|,|y|\leq 2$, yields a missing rule.
\section{Model Checking}
Model checking is a central decision problem in every logic.
It asks the simple question whether a formula of the logic is true in a structure.
One often considers the \emph{data complexity} variant where the formula is fixed and the input consists of only the structure.
For dependence logic, model checking is $\NP$-complete in general, even for quantifier-free disjunctions of only three dependence atoms, but in $\NL$, when restricted to disjunctions of two atoms \cite{DBLP:journals/sLogica/Kontinen13}.
% Kontinen \cite{DBLP:journals/sLogica/Kontinen13} also showed complexity
However, with approximations even simpler formulas can have intractable model checking, as was shown for the approximate operator in \cite{DBLP:journals/amai/DurandHKMV18}.
Our first result demonstrates that with a single approximate dependence atom, disjunctions of just two atoms become $\NP$-complete.
In the setting of dependence logic, disjunctions are not defined globally over the whole team but rather split the team into two parts, each satisfying one of the two disjuncts.
Formally, $X \models \phi_1 \lor \phi_2$ if there are teams $X_1, X_2$ such that $X = X_1 \cup X_2$ and $X_i \models \phi_i$ for $i \in \{1,2\}$.

\begin{restatable}{thm}{MCuvxy}\label{thm:disjunction NP-complete}
    The model checking problem for $\dep_{\frac{3}{13}}(u, v) \lor \dep_{0}(x, y)$ is $\NP$-complete.    
\end{restatable}

Next, consider the complexity of a single approximate dependence atom.
It is well known that single dependence atoms are first-order definable.
This is no longer true for approximate dependence atoms.

\begin{restatable}{thm}{MCxy}
    The model checking problem for $\dep_{\frac{1}{2}}(x, y)$ is $\LOGSPACE$-hard.
\end{restatable}

\section{Conclusion and Future Work}

%--- More general completeness proof

%--- Monoid-setting (+++ Minna writes)

%--- Other ways to measure the approximation (quantity, local variant etc.)

%~\\~\\

%-- Is it so that the problem is not finitely axiomatizable in the general case?

%\Cref{missing} 
We introduced a semantic entailment that cannot be deduced by using system $\mathbf{A}^*$, meaning that the system is not complete in the general case of approximate dependencies. It seems likely that there are other similar semantic entailments involving different approximations,
%that cannot be deduced using the system, 
so the system might actually have several missing rules. This raises the question whether the implication problem is finitely axiomatisable or whether every complete system necessarily involves an infinite number of rules. If the problem is not finitely axiomatisable, a natural goal for future work would be to prove that this is indeed the case. It might still be possible to represent the axioms in a nice way that expresses the infinite rules in the form of a rule schema, e.g., as in \cite{10.1145/78935.78937,10.1145/3677120,barlag2026locallyconsistentkrelationsentailment}. Investigating this is also a relevant question for further research.

Since the implication problem for dependencies with real (or even rational) number approximations seems difficult to axiomatise, one may ask if this has something to do with some specific properties of these structures. Therefore, it might be useful to investigate the approximate dependence in the setting of monoids, where we may consider structures with different properties. A monoid is a tuple $K=(K,+,0_K)$ where $K$ is a set, $+$ is an associative binary operation on $K$, and $0_K$ is an identity element of $+$. A monoid $K$ is \emph{positive} if $a+b=0_K$ implies $a=0_K$ and $b=0_K$, and \emph{commutative} if $a+b=b+a$ for all $a,b\in K$. An \emph{ordered} 
%(commutative)
monoid is a commutative monoid $K$ with a partial order $\leq$ such that $0_K\leq a$ for all $a\in K$, and $a\leq b$ implies $a+c\leq b+c$ for all $a,b,c\in K$.
The idea is that we consider a totally ordered positive commutative monoid $K=(K,\leq,+,0_K)$, and for a finite team $X\subseteq\{s\mid s\colon D\to M\}$, define a $K$-team $\mathbb{X}$ as a function $\mathbb{X}\colon X\to K$. A $K$-team is essentially an annotated database relation or a $K$-relation, as introduced in \cite{green07} in the setting of semiring provenance. Denote $\text{supp}(\mathbb{X}):=\{s\in X\mid \mathbb{X}(s)\neq0_K\}$.
%, and for any $Y\subseteq X$, the $K$-team $\text{rem}_Y(\mathbb{X})$ such that $\text{rem}_Y(\mathbb{X})(s)=0_K$ if $s\in Y$ and $\text{rem}_Y(\mathbb{X})(s)=\mathbb{X}(s)$ otherwise. 
%Let $a\in K$.
Then for any $a\in K$, we can define a $K$-approximate dependence atom $\dep_a(x,y)$ with the following satisfaction relation: $\mathbb{X}\models \dep_a(x,y)$ iff there is $Y\subseteq X$ for which $\sum_{s\in Y}\mathbb{X}(s)\leq a$ and 
%the $K$-team $\mathbb{X}':=\text{rem}_Y(\mathbb{X})$
%$\mathbb{X}'(s)=\begin{cases}
%    0_K, \text{ if } s\in Y\\
%    \mathbb{X}(s), \text{ otherwise}
%\end{cases}$
$\text{supp}(\mathbb{X}\restriction_{X\setminus Y})\models\dep(x,y)$. Note that we define $\sum_{s\in \emptyset}\mathbb{X}(s)=0_K$, so $\mathbb{X}\models\dep_{0_K}(x,y)$ iff $\text{supp}(\mathbb{X})\models\dep(x,y)$.
If $K$ is the Boolean monoid $\mathbb{B}=(\{0,1\},\leq,\lor,0)$, the $K$-approximate dependencies correspond to the usual dependencies, where each $\dep_0(x,y)$ corresponds to $\dep(x,y)$ and each $\dep_1(x,y)$ is trivially satisfied. If $K$ is the probability monoid $\mathbb{P}=([0,1],\leq,+,0)$, where $\leq$ is the usual order of the unit interval $[0,1]$ and $a+b=\max\{a+b,1\}$, then the $K$-approximate dependencies correspond to the approximate dependencies over $K$-teams $\mathbb{X}$ that are uniform distributions, i.e., $\mathbb{X}(s)=1/|X|$ for all $s\in X$. Different dependency notions over $K$-teams have been studied, e.g., in \cite{hannula24,hirvonen26}\footnote{Note that these works also consider (conditional) independence atoms whose semantics is defined using both addition and multiplication, so their structure $K$ is a \emph{semiring} instead of a monoid.}.

%In \Cref{monoids}, 
In the setting of monoids, the satisfaction of the $K$-approximate dependencies is defined in such a way that the number of tuples that can be removed is measured as an absolute value instead of a team ratio. %In addition to the obvious topic of considering the implication problem in the monoid setting, 
One possible line of future research 
%could be to continue focusing on the case of 
on real approximations is to examine whether considering approximation as an absolute value instead of a ratio could help with obtaining a complete axiomatisation. Already in \cite{KIVINEN1995129}, variants of both absolute value approximations and ratio approximations were introduced. One can also consider other definitions, such as a \emph{local} variant where $X\models {\dep_p(x, y)}$ if there is a subteam $Y\subseteq X$, $|Y | \leq p \cdot |X\upharpoonright_{\{xy\}}|$, such that $X \setminus Y\models \dep(x, y)$, where $X\upharpoonright_{\{xy\}}$ is the team restricted to the variables $x$ and $y$. With this definition, $\{\dep_p(x,w), \dep_q(w,y)\}\not\models \dep_{p+q}(x,y)$, meaning any complete system, even restricted to unary atoms, would be different from the one in \cite{Väänänen2017}.

To take advantage of the team semantic framework, we could situate the approximate dependence atom in a language with connectives and quantifiers as in first-order dependence logic \cite{vaananen2007dependence}. Some first steps in this direction are taken in \cite{DBLP:journals/amai/DurandHKMV18}. With our gained understanding of unary and 1-constancy approximate dependence atoms, they could serve as the first building blocks of an approximate variant of first-order dependence logic.

% ++change logical implication to semantic entailment (or other way around)
% %\minna{Added something about the quantity variant. We might want to say something more about this (multiteams)?}

% ~\\~\\

% --- Other ways to measure the approximation (quantity, local variant etc.) \matilda{I'll add something here}

% ~\\~\\

\subsubsection*{Acknowledgements}

The authors want to thank Jouko Väänänen and the anonymous reviewers of LINDA 2026 for many helpful suggestions. The authors appreciate the project support by a Finnish-German co-operation Grant:
DAAD (German Exchange Service, project id 57710940), Research Council of Finland (project id 359650).
The first author appreciates funding by the German Research Foundation (DFG)
under the grant ME2479/3-1 and project id 511769688.
The fourth author appreciates funding by the Magnus Ehrnrooth Foundation.

\bibliography{bibfile}

\newpage

\appendix
\section{Proofs}
The appendix includes the proofs omitted from the main part of the paper.

%\asa{Should these come in numerical order?}
%\minna{Fixed.}

\lemMissing*
\begin{table}[t]
    \caption{Types of error for $v_1, v_2$ and $v_3$. We denote by $\bullet$ any non zero value.}
    \label{tab:bad assignments}
    \[
    \begin{array}{ccccc}
    \toprule
        \text{type} & x & v_1 & v_2 & v_3 \\ 
    \midrule
        \mathcal{A} & 0 & \bullet & 0 & 0 \\
        \mathcal{B} & 0 & \bullet & \bullet & 0 \\
        \mathcal{C} & 0 & \bullet & \bullet & \bullet \\
        \mathcal{D} & 0 & \bullet & 0 & \bullet \\
        \mathcal{E} & 0 & 0 & 0 & \bullet \\
        \mathcal{F} & 0 & 0 & \bullet & \bullet \\
        \mathcal{G} & 0 & 0 & \bullet & 0 \\
    \bottomrule
        % h & 0 & 0 & 0 & 0 & 
    \end{array}
    \]
\end{table}
\begin{proof}
    Fix any team $X \models \Sigma$.
    Notice that there are seven different types of `error' assignments in $X$ for variables $v_1, v_2$ and $v_3$, i.\,e., assignments which need to be removed to make $\dep_{\frac{1}{8}}(x, v_1v_2)$ or $\dep_{\frac{2}{8}}(x, v_1v_3)$ true.
    These `error types' are depicted in \Cref{tab:bad assignments} and labelled $\mathcal{A}$ to $\mathcal{G}$.
    Note that each type corresponds to a non-empty subset of the set $\{v_1, v_2, v_3\}$.
    We denote by $a$ the ratio of error type $\mathcal{A}$ occurring in $X$, i.\,e., 
    \[
        a = \frac{|\{s \in X \mid s \text{ is of error type } \mathcal{A}\}|}{|X|}.
    \]
    Likewise define the ratios $b$ to $g$.    

    We have two constraints on these ratios given by $\Sigma$:
    \[
        a+b+c+d+f+g \leq \frac{1}{8}
    \]
    the ratio of assignments which have to be removed for $\dep_{\frac{1}{8}}(x, v_1v_2)$, and
    \[
        a+b+c+d+e+f \leq \frac{2}{8}.    
    \]
    the ratio of assignments which have to be removed for $\dep_{\frac{2}{8}}(x, v_1v_3)$.

    We now show, that $p = \frac{5}{16}$ is the smallest ratio to guarantee $X \models \dep_p(x, y)$.
    Here, there are two ways to show $X \models \dep_p(x, y)$.
    First, via $\dep_q(x,v_2v_3)$ and $\dep_0(v_2v_3, y)$ with unknown ratio $q = b+c+d+e+f+g$.
    Second, via $\dep_r(x, v_1)$ and $\dep_{\frac{2}{8}}(v_1, y)$ with unknown ratio $r = a+b+c+d$.
    With transitivity it follows that $p = \min(q + 0, r + \frac{2}{8})$.
    We need to show that $p \leq \frac{5}{16}$ for any $X$, i.\,e., we want to show the maximum value $p$ can take is less than or equal to $\frac{5}{16}$.
    Thus, we assume $q = r + \frac{2}{8}$ as the worst case for $p$ and try to find the maximum.
    If we plug in the lower bounds we get
    \[
        b+c+d+e+f+g = a+b+c+d + \frac{2}{8}
    \]
    which implies $e+f+g = \frac{2}{8} + a$.
    From the first constraint, it follows that $f+g \leq \frac{1}{8} - r$.
    Thus $e \geq \frac{1}{8} + r + a$.
    From the second constraint, it follows that $e \leq \frac{2}{8} - r - f$.
    Together we have
    \begin{align*}
        \frac{2}{8} - r - f &\geq \frac{1}{8} + r + a \\
        \frac{1}{8} - a - f &\geq 2r \\
        \frac{1}{16} - \frac{a + f}{2} &\geq r \\
    \end{align*}
    Since we want to maximize $r$, assume $a = f = 0$ and $r = \frac{1}{16}$.
    So we get $p = \frac{2}{8} + \frac{1}{16} = \frac{5}{16}$ as desired.
    Therefore any team $X$ that satisfies $\Sigma$ must also satisfy $\dep_{\frac{5}{16}(x, y)}$.
\end{proof}

\thmUnary*

 \begin{proof}
 Let $D_0$ be the set of variables in $\Sigma\cup \{\dep_p(x,y)\}$, and let $z_0$ be a variable that does not occur in any of the atoms in the set. We write $\vdash$ for a derivation using rules $A1^1$, $A2^1$, $A3^1$, $A6^1$ and $A7^1$.
 Suppose that $\Sigma\not\vdash \dep_p(x,y)$. By $A1^1$ and $A7^1$, we can derive $\vdash \dep_0(x,x)\vdash \dep_p(x,x)$ and by $A2^1$ we can derive $\vdash \dep_1(x,y)$, so it follows that $y$ is different from $x$ and that $p<1$. 
 
 For each variable $u$ in $D_0$, define $d(u)$ to be the smallest rational number for which $\Sigma\vdash \dep_{d(u)}(x,u)$, which is guaranteed to exist since $\Sigma$ is finite and $\vdash \dep_1(x,u)$ always holds by $A2^1$ (also in the case that $x$ is the empty sequence).
 
Let $n-1$ be the least common denominator for the approximations appearing in $\Sigma\cup \{\dep_p(x,y)\}$. For each approximation $q$, let $q^\prime$ be such that $q=\frac{q'}{n-1}$, and similarly, we write $d(u)'$ whenever $d(u)=\frac{d(u)'}{n-1}$, etc.
We build the counterexample team $X=\{s_0,\dots s_{n-1}\}$ such that for each variable $w\in D_0$, 
 \[   
s_i(w) = 
     \begin{cases}
       i &\quad\text{if } i\leq d(w)',\\
       d(w)' &\quad \text{otherwise.} \\
     \end{cases}
\]
Additionally, let $s_i(z_0)=i$ for all $i\in\{0,\dots,n-1\}$.

 %Additionally, the team is constructed such that for $w\in Var$, $X[w]=d(w)'+1$.
%
%
% Before showing that the counterexample team $X$ is as intended, we prove the following claim: For any atom $\dep_q(u,v)$,  $$X\models \dep_q(u,v)  \text{ iff } |X[v]\setminus X[u]|\leq qn.$$
% %We have that $X\models \dep_q(u,v)$ if we have to remove at most $qn$ assignments from $X$ to satisfy the atom $\dep(u,v)$. 
%
The team is constructed such that for $w\in D_0$, $|X[w]|=d(w)'+1$, see \Cref{R team} for an example. As a consequence, we have the equivalences:
%By construction of $X$, $X[v]=d(v)'+1$ and $X[u]=d(u)'+1$. We then obtain the following equivalences: 
%
% $$\quad\quad\quad\quad X\models \dep_q(u,v)  \text{ iff }  d(v)'- d(u)'\leq qn  \text{ iff } |X[v]\setminus X[u]|\leq qn. \quad\quad\quad\quad (*)$$

\begin{equation*}
    X\models \dep_q(u,v)  \text{ iff }  d(v)'- d(u)'\leq qn  \text{ iff } |X[v]\setminus X[u]|\leq qn. \tag{$*$}
\end{equation*}

Now, let us show that $X\not\models \dep_p(x,y)$. Observe that $\Sigma\vdash \dep_0(x,x)$ by $A1^1$, so $X[x]=\{0\}$ if $x$ is a nonempty sequence and otherwise $\emptyset$, so $|X[x]|\leq 1$. Furthermore, since $\Sigma\not\vdash \dep_p(x,y)$, we have by construction of $X$ that $|X[y]|\geq p'+2$, hence $|X[y]\setminus X[x]|\geq p'+1$. Now,  $$\frac{|X[y]\setminus X[x]|}{n}\geq\frac{p'+1}{n} > \frac{p'}{n-1}=p,$$ thus $X\not\models \dep_p(x,y)$ follows by ($*$). 

It remains to show that for each $\dep_r(u,v)\in\Sigma$, $X\models \dep_r(u,v)$. 
By construction of the team, $|X[u]\setminus X[x]|=d(u)'$. By rule $A6^1$ we have that $\Sigma\vdash \dep_{d(u)+r}(x,v)$, hence $|X[v]\setminus X[x]|\leq d(u)'+r'$. Now $$|X[v]\setminus X[u]|\leq (d(u)'+r')-d(u)'=r'=r(n-1)< rn,$$ from which $X\models \dep_r(u,v)$ follows by ($*$).  

If we want to avoid the dummy variable $z_0$, identify the smallest $k$ for which $s_k\in X$ and $s_k$, $s_l$ are the same assignments restricted to the variables in $D_0$ whenever $k< l\leq n-1$. For each such $s_l$, replace it with the assignment $s_l^a$ such that $s_l^a(w)=a_l$ for all $w\in D_0$, where $a_l$ is a fresh value. Let $Y$ be the team $\{s_0, \dots, s_{k},s_{k+1}^a, \dots s_{n-1}^a\}$. Checking that $Y$ is a counterexample team is similar to the proof for $X$. 
\end{proof}

\MCuvxy*
\begin{proof}
    Membership in $\NP$ is quite obvious.
    A non-deterministic machine can guess the split into $X_1$ and $X_2$ for the disjunction and additionally guess the removed $\frac{3}{13}$ assignments $Y_1 \subseteq X_1$.
    Then the only thing left to do is to check whether $X_1 \setminus Y_1 \models \dep(u, v)$ and $X_2 \models \dep(x, y)$ is true.
    
    We continue with showing $\NP$-hardness.
    To this end, we present a reduction from $\ThreeSAT$.
    This reduction will use a similar construction as the reduction to the decision problem of $\MaxTwoSAT$ presented in \cite{DBLP:journals/tcs/GareyJS76}.

    Let $\phi$ be an instance of $\ThreeSAT$, i.\,e., a set of clauses $(a_i \lor b_i \lor c_i)$ for $i \in I$, where each $a_i, b_i$ and $c_i$ corresponds to a variable or its negation.
    First, split each clause $(a_i \lor b_i \lor c_i)$ into four subclauses of size 1 and six subclauses of size 2:
    \[
        (a_i), (b_i), (c_i), (d_i), (\bar a_i \lor \bar b_i), (\bar a_i \lor \bar c_i), (\bar b_i \lor \bar b_i), 
        (a_i \lor \bar d_i), (b_i \lor \bar d_i), (c_i \lor \bar d_i).
    \]
    The variable $d_i$ is a fresh variable introduced for each clause.
    As noted in \cite{DBLP:journals/tcs/GareyJS76} exactly seven of these ten clauses can be true when the original clause was satisfied, while at most six can be true if $a_i, b_i$ and $c_i$ are false.

    We now construct a team $X_i$ such that the clause is satisfied if and only if a split between 17 of 20 assignments is possible.
    Let $p\colon \operatorname{Lit} \to \{0, 1\}$ be the function determining the parity of literals; while $\bar p$ denotes the opposite, i.e., $\bar p(x) = 1- p(x)$ for all $x \in \operatorname{Lit}$.
    Now, the team $X_i$ is as depicted in \Cref{fig:team X_i}.
    The full team of this reduction is then the union of all $X_i$ for $i \in I$, i.e. $X = \bigcup_{i \in I} X_i$.

    \begin{table}[b]
        \caption{Team $X_i$ given a clause over variables $\{a,b,c\}$, with $p(x)$ the parity of $x$ in the clause; and $\bar p$ the opposite.}
        \label{fig:team X_i}
        \[
            \begin{array}[t]{cccc}
            \toprule
                u & v & x & y \\
            \midrule
                i_0 & 0 & a & p(a) \\
                i_0 & 1 & a & p(a) \\
                i_1 & 0 & b & p(b) \\
                i_1 & 1 & b & p(b) \\
                i_2 & 0 & c & p(c) \\
                i_2 & 1 & c & p(c) \\
                i_3 & 0 & d_i & 1 \\
                i_3 & 1 & d_i & 1 \\
            \bottomrule
            \end{array}
            \qquad
            \begin{array}[t]{cccc}
            \toprule
                u & v & x & y \\
            \midrule
                i_4 & 0 & a & \bar p(a) \\
                i_4 & 1 & b & \bar p(b) \\
                i_5 & 0 & a & \bar p(a) \\
                i_5 & 1 & c & \bar p(c) \\
                i_6 & 0 & b & \bar p(b) \\
                i_6 & 1 & c & \bar p(c) \\
            \bottomrule
            \end{array}
            \qquad
            \begin{array}[t]{cccc}
            \toprule
                u & v & x & y \\
            \midrule
                i_7 & 0 & a & p(a) \\
                i_7 & 1 & d_i & 0 \\
                i_8 & 0 & b & p(b) \\
                i_8 & 1 & d_i & 0 \\
                i_9 & 0 & c & p(c) \\
                i_9 & 1 & d_i & 0 \\
            \bottomrule
            \end{array}
        \]
    \end{table}
    
    Correctness follows from \cite{DBLP:journals/tcs/GareyJS76} in a mostly straightforward manner.
    First, observe that $\dep_0(x, y)$ corresponds to assignments of the variables.
    That is, in the split $X_2 \models \dep_0(x, y)$, each variable has to agree on its parity.
    Second, the split of $\dep_{\frac{3}{13}}(u, v)$ can take one assignment from each subclause without considering the approximation.
    Taking the approximation into account, this side of the split can take three additional assignments for a total of 13.
    Ordinarily, these three assignments would violate the dependence atom; however, they can be removed using the approximation.
    Thus, if seven of the ten subclauses are true by an assignment team, $X_i$ can be split.
    If this is true for all $X_i$, then this also holds for $X$.
    In contrast, if only six assignments are true (as is the case for the non-satisfying assignment of the clause), then the team $X_i$ cannot be split, and therefore also  $X$ cannot be split.
\end{proof}

\MCxy*
\begin{proof}
    Define $X^\leftrightarrows$ as the team that swaps and renames the values of $x$ and $y$ in every assignment of $X$.
    For example, if $s(x) = 0, s(y) = 1$ is an assignment in $X$, then $X^\leftrightarrows$ contains $\bar s(x) = \bar 1, \bar s(y) = \bar 0$ as an assignment.
    %, i.e., $Y = \{s(yx) \mid s(xy) \in X\}$
    
    Let $X$ be a team over $\{x, y\}$.    
    We reduce from the model checking problem for $\dep(x, y) \lor \dep(y, x)$, which is $\LOGSPACE$-complete \cite{DBLP:conf/csl/0001KM26}, by mapping $X$ to $X \cup X^\leftrightarrows$.
    It is easy to see that this can be done by a first-order reduction.

    Assume $X \models \dep(x, y) \lor \dep(y, x)$.
    Then there exists a split $X = X_1 \cup X_2$ such that $X_1 \models \dep(x, y)$ and $X_2 \models \dep(y, x)$. 
    The same split also witnesses satisfaction for $X^\leftrightarrows$, when swapping the teams, i.e., $X^\leftrightarrows_2 \models \dep(x, y)$ and $X^\leftrightarrows_1 \models \dep(y, x)$.
    Since the values in $X$ and $X^\leftrightarrows$ are disjoint $X_1 \cup X^\leftrightarrows_2 \models \dep(x, y)$.
    Now, for the size of these subteams we have that
    \begin{align*}
        |X_1| + |X_2| &\geq |X| \\
        |X_1| + |X^\leftrightarrows_2| &\geq |X| \tag{$|X_2| = |X^\leftrightarrows_2|$} \\
        |X_1 \cup X^\leftrightarrows_2| &\geq |X| \tag{$X_1 \cap X^\leftrightarrows_2 = \emptyset$} \\
        |X_1 \cup X^\leftrightarrows_2| &\geq \frac{1}{2} \cdot (|X| + |X^\leftrightarrows|) \tag{$|X| = |X^\leftrightarrows|$} \\
        |X_1 \cup X^\leftrightarrows_2| &\geq \frac{1}{2} \cdot |X \cup X^\leftrightarrows| \tag{$X \cap X^\leftrightarrows = \emptyset$}
    \end{align*}    
    Therefore $X \cup X^\leftrightarrows \models \dep_{\frac{1}{2}}(x, y)$.

    For the other direction assume $X \not\models \dep(x, y) \lor \dep(y, x)$.
    Then for all subteams $X_1, X_2 \subseteq X$ with $X_1 \models \dep(x, y)$ and $X_2 \models \dep(y, x)$ we have that $X \setminus (X_1 \cup X_2) \neq \emptyset$.
    Therefore $|X_1| + |X_2| < |X|$.
    The same is true for $X^\leftrightarrows, X^\leftrightarrows_1$ and $X^\leftrightarrows_2$.
    Thus
    \begin{align*}
        |X_1| + |X_2| + |X^\leftrightarrows_1| + |X^\leftrightarrows_2| &< |X| + |X^\leftrightarrows| \\ 
        2|X_1| + 2|X^\leftrightarrows_2| &< |X| + |X^\leftrightarrows| \\ 
        |X_1| + |X^\leftrightarrows_2| &< \frac{1}{2}(|X| + |X^\leftrightarrows|) \\ 
        |X_1 \cup X^\leftrightarrows_2| &< \frac{1}{2}|X \cup X^\leftrightarrows|
    \end{align*}
    holds and implies $X \cup X^\leftrightarrows \not\models \dep_{\frac{1}{2}}(x, y)$.
    % For the other direction assume $X \cup X^\leftrightarrows \models \dep_{\frac{1}{2}}(x, y)$.
    % Then there exists $X_1 \cup X^\leftrightarrows_2 \models \dep(x, y)$ with $|X_1 \cup X^\leftrightarrows_2| \geq \frac{1}{2} \cdot |X \cup X^\leftrightarrows|$ as before.
    % Thus we know $X_1 \models \dep(x,y )$ and $X_2 \models \dep(y, x)$ and most importantly $X_1 \cup X_2 = X$.
    % Therefore $X \models \dep(x, y) \lor \dep(y, x)$.
\end{proof}

\end{document}